\documentclass[12pt]{article}
\usepackage{graphicx}
\usepackage{amssymb}
\usepackage[all]{xy}
\usepackage{color}
\usepackage{comment}

\usepackage{amsmath}
\usepackage{amssymb}
\usepackage{stackengine}
\usepackage{scalerel}
\usepackage{cite}

\usepackage[latin1]{inputenc}

\newcommand{\newc}{\newcommand}
\newc{\beq}{\begin{equation}}
\newc{\eeq}{\end{equation}}
\newc{\bea}{\begin{array}}
\newc{\eea}{\end{array}}
\newcommand{\ben}{\begin{eqnarray}}
\newcommand{\een}{\end{eqnarray}}
\newc{\ra}{\rightarrow}
\newc{\bfx}{{\bf x}}
\newc{\bfV}{{\bf V}}
\newc{\cO}{{\cal O}}
\newc{\bfv}{{\bf v}}
\newc{\bfu}{{\bf u}}
\newc{\bfp}{{\bf p}}
\newc{\ve}{{\varepsilon}}
\newc{\Psibar}{\overline\Psi}
\newc{\w}{{\bf w}}
\newc{\E}{{\mathbf{E}}}
\newc{\EE}{{\mathcal E}}
\newc{\bfn}{{\mathbf\nabla}}
\newc{\la}{{\cal L}}
\newc{\tla}{{\tilde{\cal L}}}
\newc{\bp}{{\bf p}}
\newc{\ho}{\hookrightarrow }
\newc{\bP}{{\bf P}}
\newc{\pd}{{\partial}}
\newc{\piv}{{\partial_4}}
\newc{\pv}{{\partial_5}}
\newc{\bJ}{{\bf J}}
\newc{\bze}{{\mathbf 0}}
\newc{\bK}{{\bf K}}
\newc{\tphi}{{\tilde\phi}}
\newc{\tF}{{\tilde F}}
\newc{\tD}{{\tilde D}}
\newc{\tJ}{{\tilde J}}
\newc{\tj}{{\tilde j}}
\newc{\bD}{{\bf D}}
\newc{\tvphi}{{\tilde\varphi}}
\newc{\trho}{{\tilde\rho}}
\newc{\ttheta}{{\tilde\theta}}
\newc{\tpsi}{{\tilde\psi}}
\newc{\tu}{{\tilde u}}
\newc{\cD}{{\cal D}}
\newc{\tPhi}{{\tilde\Phi}}
\newc{\tPsi}{{\tilde\Psi}}
\newc{\tA}{{\tilde A}}
\newc{\talpha}{{\tilde\alpha}}
\newc{\tbeta}{{\tilde\beta}}
\newc{\bA}{{\mathbf A}}
\newc{\bB}{{\bf B}}
\newc{\br}{{\bf r}}
\newc{\sig}{{\mathbf\sigma}}
\newc{\eg}{{\rm e.g.\ }}
\newc{\ie}{{\rm i.e.\ }}
\newcommand{\bey}{\begin{eqnarray}}
\newcommand{\pslash}{\not{\hbox{\kern-2.3pt $p$}}}
\newcommand{\pdslash}{\not{\hbox{\kern-2pt $\partial$}}}
\newcommand{\eey}{\end{eqnarray}}
\newtheorem{theorem}{Theorem}

\newenvironment{proof}[1][Proof]{\noindent\textbf{#1.} }{\ \rule{0.5em}{0.5em}}

\newtheorem{defn}{Definition}

\begin{document}

\begin{titlepage}
\vskip 2cm
\begin{center}
{\Large  A general formulation based on algebraic spinors for the quantum computation
%\footnote{{\tt matrindade@uneb.br}}  
{\let\thefootnote\relax\footnote{{$^{\star}${\tt matrindade@uneb.br}}}}
}
 \vskip 10pt
{  Marco A. S. Trindade$^{1 \star}$, S. Floquet$^{2}$ and J.D.M. Vianna$^{3,4}$  \\}
\vskip 5pt
{\sl $^1$Colegiado de F\'isica, Departamento de Ci\^encias Exatas e da Terra, Universidade do Estado da Bahia,
 Brazil\\
$^2$Colegiado de Engenharia Civil, Universidade Federal do Vale do S\~ao Francisco, Brazil \\
$^3$Instituto de F\'isica, Universidade Federal da Bahia, Brazil \\
$^4$Instituto de F\'isica, Universidade de Bras\'ilia, Brazil}
\vskip 2pt

\end{center}

\renewcommand\abstractname{Abstract}

\begin{abstract}

In this work we explore the structure of Clifford algebras and the representations of the algebraic spinors in quantum information theory.
Initially we present an general formulation through elements of left minimal ideals in tensor products of the Clifford algebra $Cl^{+}_{1,3}$. Posteriorly we perform some applications in quantum computation: qubits, entangled states, quantum gates, representations of the braid group, quantum teleportation, Majorana operators and supersymmetry.
Finally, we discuss advantages related to standard Hilbert space formulation.
\end{abstract}

\bigskip



\vskip 3pt

\end{titlepage}

%\tableofcontents

\newpage

\setcounter{footnote}{0} \setcounter{page}{1} \setcounter{section}{0} %
\setcounter{subsection}{0} \setcounter{subsubsection}{0}
%%%%%%%%%%%%%%%%%%%%%%%%%%%%%%%%%%%%%%%%%%%%%%%%%%%%%%%%%%%%%%%%%%%%%%%%%%
%%%%%%%%%%%%%%%%%%%%%%%%%%%%%%%%%%%%%%%%%%%%

\section{Introduction}

The theory of quantum information and quantum computation \cite {Nielsen} has aroused great interest because of the possibility of the quantum computer to perform tasks that are not feasible for a classical computer \cite{Preskill}. Such theory is usually formulated in Hilbert's space language; notwithstanding alternative formulations and techniques based on Clifford algebras have been proposed \cite{Havel1, Havel2, Baylis}. Advantages can be perceived through the possibility of a clear geometric interpretation of the concepts involved and the unified in a natural language for the description of multipartite quantum systems. By way of example, Vlasov \cite{Vlasov} has shown that Clifford algebras can be used to construct universal sets of quantum logic gates. The method explores the relationship between Clifford algebras and Lie $ u (n) $ algebras of the unitary group. We can notice explicitly that this proposal is more advantageous than other matrix approaches, when the Hamiltonian of the system is known, which corresponds to the usual cases. Cabrera \emph{et al} \cite{Cabrera} have shown that Clifford algebras enable the elaboration of more efficient algorithms for quantum state engineering also by appropriating the relationships between Clifford and Lie algebras. It is noteworthy that \cite{Cabrera}, unlike the above, q-bits are described in purely algebraic terms, providing the operational efficiency of the algorithm. In \cite{Alves} Alves and Lavor showed that the use of Clifford algebra in the Grover algorithm simplifies the computation of its computational complexity. Josza and Miyake \cite{Josza} used the Jordan-Wigner representation of a Clifford algebra to obtain efficient classical simulations of quantum circuits, called Gaussian quantum circuits. Clifford algebras have been employed for the implementation of the quantum Fourier transform facilitating the description of a sequence of pulses associated with a given quantum logic gate \cite{Wein}. In the reference \cite{Havel1} operatorial spinors are used in the construction of qubits and various nuances of quantum information processing within the scope of nuclear magnetic resonance are analyzed, including decoherence. \

A plethora of theoretical proposals has been launched  with the aim of obtaining more realistic models that are able to mitigate decoherence\cite{Nielsen, Fujji}, one of the biggest obstacles to the experimental implementation of a quantum computer. One of the most compelling proposals is topological quantum computation \cite{topological1, topological2, topological3, topological4} whose experimental implementation is related to particles called anyons \cite{topological5}, which have an intermediate statistic between bosons and fermions. In this context, particles called Majorana fermions have been used in the experimental proposals. Its description in Hilbert's standard space formulation involves the use of Clifford algebras with the operators of Majorana \cite{Lee, Yu,Lombardo}.

 In this work our focus is under  perspective of  $Cl_{1,3}\simeq Cl_{3,0}$ algebra using the tensor product of algebras and minimal lefts.
Unlike the above references, we use algebraic \cite{Havel1, Havel2} (and non-operatorial) spinors in the Clifford algebra $ Cl_ {1,3} ^ {+} $ as the guiding thread for the description of quantum information. We explore the connections between braid group as well aspects related to supersymmetry, that have implications in the Hamiltonian properties.  Firstly, in \ref{sec2} section we present some basic concepts concerning Clifford's theory of algebras. In section \ref{sec3} , we release a general algebraic formulation for multipartite q-bit systems in terms of the algebraic spinors in $ Cl_ {1,3} $. In section \ref{sec4}, topics such as braid group and quantum teleportation were covered. In section \ref{sec5}, we show how Majorana operators can be constructed in the $ Cl_ {1,3} ^ {+} $ algebra. We present a supersymmetry algebra structure $ N = 2 $ and show how the existence of this algebra implies the degeneracy of the eigenstates of the Hamiltonian. Finally, in section \ref{sec6}, we have the conclusions and perspectives.

\section{Preliminaries} \label{sec2}

In this section, we review some basic notions about Clifford algebras, by using reference \cite{Vaz}. \

 Let $V$ a space vector over $\mathbb{R}$ equipped with symmetric bilinear form $g$ and $A$ an associative algebra with unit $1_{A}$ and let $\gamma$ be the linear application $\gamma:V\rightarrow A$. The pair $(A, \gamma)$ is a Clifford algebra for the quadratic space $(V,g)$ when $A$ is generated as an algebra by $\{\gamma(v); v \in V\}$, $\{a 1_{A}; a \in R\}$ and satisfy
\begin{eqnarray}
\gamma(v)\gamma(u)+\gamma(u)\gamma(v)=2g(v,u)1_{A},
\end{eqnarray}
for all $v,u \in V$. Consider $V$ the vector space $\mathbb{R}^{n}$ and $g$ a symmetric bilinear form in $\mathbb{R}^{n}$ of signature $(p,q)$ with $p+q=n$.
We called this quadratic space by $Cl_{p,q}=Cl(R^{p,q})$. The even subalgebra is defined by:
\begin{eqnarray}
Cl^{+}_{p,q}=\{\Gamma \in Cl_{p,q} ; \Gamma = \# \Gamma= \widehat{\Gamma} = (-1)^{p}\Gamma_{p}\},
\end{eqnarray}
 where $\#()$ and $\widehat{\ . \ }$ denote graded involution, that keep the sign of the elements belonging to even subspaces. % em que $\Gamma_{[k]}=(-1)\Gamma_{[k]}$
Other definitions used refer to groups $spin(p,q)$ and $spin_{+}(p,q)$, given by:
\begin{eqnarray}
spin(p,q)=\{a \in Cl_{p,q}^{+}; N(a)=\pm 1 \}
\end{eqnarray}
and
\begin{eqnarray}
spin_{+}(p,q)=\{a \in Cl_{p,q}^{+}; N(a)= 1 \},  \label{eq4}
\end{eqnarray}
respectively, wherein $N(a)=\mid a \mid ^{2}=\langle \widetilde{a} a \rangle_{0}$ is related to norm of elements of Clifford algebra,  $\widetilde{\ . \ }$ represent the reversion operator defined by
$\widetilde{A}_{[k]}=(-1)^{k(k-1)/2}A_{[k]};$   $\langle \ \ \rangle_{k}:Cl_{p,q}\rightarrow \bigwedge_{k}(\mathbb{R}_{p,q})$, with $\bigwedge_{k}(\mathbb{R}_{p,q})$ is the exterior algebra of vector space $\mathbb{R}_{p,q}$. \\

In this article, we will use Clifford algebra $Cl_ {1,3}$, commonly called spacetime algebra, defined by:
\begin{eqnarray}
\gamma_{\mu}\gamma_{\nu}+\gamma_{\nu}\gamma_{\mu}=2g_{\mu \nu},
\end{eqnarray}
 where $g_{\mu \nu}$ is the Minkowiski metric with signature $(+---)$. Consequently, $\gamma_{\mu}\gamma_{\nu}=-\gamma_{\nu}\gamma_{\mu}, \mu \neq \nu$ and $\gamma_{0}^{2}=1, \gamma_{1}^{2}=\gamma_{2}^{2}= \gamma_{3}^{2}=-1$. An arbitrary element $\Gamma \in Cl_{1,3}$ can be written as a linear combination of the vectors belonging to the bases $\{1\}$ (scalar), $\{\gamma_{0}, \gamma_{1}, \gamma_{2}, \gamma_{3}\}$ (vectores), $\{\gamma_{0}\gamma_{1}, \gamma_{0}\gamma_{2}, \gamma_{0}\gamma_{3}, \gamma_{1}\gamma_{2}, \gamma_{1}\gamma_{3}, \gamma_{2}\gamma_{3}\}$, (bivectors),
$\{\iota \gamma_{0}, \iota \gamma_{1}, \iota \gamma_{2}, \iota \gamma_{3}\}$ (pseudovectors) and $\{\iota \equiv \gamma_{0} \gamma_{1} \gamma_{2} \gamma_{3}\}$ (pseudoscalar). \

This algebra is adequate to describe relativistic quantum mechanics without the need for imaginary numbers and with a natural geometric interpretation \cite{Hestenes}. Dirac equation can be written as \cite{Hestenes}
\begin{eqnarray}
\nabla \psi \iota \gamma_{3}\gamma_{0}=m \psi \gamma_{0},
\end{eqnarray}
where $\nabla= \gamma^{\mu}\partial_{\mu}$. Due to this fact, we will use this algebra as a starting point for future developments in relativistic quantum information theory. \

Let $A$ be an algebra. A vector subspace $I_{E}\subset A$ is a left ideal if $AI_{E} \subset I_{E}$ and similarly, a vector subspace $I_{D}\subset A$ is a right ideal if $I_{D}A \subset I_{D}$. An element $f \in A$ is an idempotent if $f^{2}=f$ and $f \neq 0$. An idempotent is primitive if there are no idempotents $f_{1}$ and $f_{2}$ such as $f_{1}f_{2}=0$ and $f_{1}+f_{2}=f$. It is possible that minimal left ideals in the Clifford algebra $Cl_{p,q}$ may be built with the product $Cl_{p, q}f$, where $f$ is a primitive idempotent. Algebraic spinors are elements of minimal left ideals. In the next section we will show how to describe quantum states in terms of algebraic spinors. Some works \cite{Havel1, Havel2} describe q-bits in terms of operatorial spinors, as highlighted in the introduction. Our choice is because we understand that the use of algebraic spinors provides a more efficient way of describing quantum mechanical structures. The transposition into usual Hilbert's space language of space is also more natural and can be easily perceived in terms of a column matrix representation of the algebra by using only a column.
%\newpage

\section{General Results}

In this section we will show some results about the universality of tensor product of Clifford algebras based on reference \cite{Vaz,chevalley} and the identification of a structure of $c^{\star}$-algebra associated to operators in $Cl_{1,3}^{+}$.

Let $V$ be a vector space over $\mathbb{R}$ with the symmetric bilinear form $g$, $C$ an unital associative algebra and $\gamma$, a linear mapping $\gamma : V \otimes V \otimes \cdots \otimes V \rightarrow C$

\begin{defn}
The pair $(C,\rho)$ is a product algebra for the quadratic space $(V \otimes V \otimes \cdots \otimes V \equiv V^{\otimes n}, g)$, when $C$ is generated by $\left\lbrace \rho(v_{1}\otimes v_{2}\otimes \cdots \otimes v_{n}); \right.$ $\left. v_{1}\otimes v_{2} \otimes \right.$ $\left. \cdots \otimes  v_{n}  \in V^{\otimes n} \right\rbrace$ and $\rho$ satisfies
\begin{eqnarray}
 \rho(v_{1}\otimes v_{2}\otimes \cdots \otimes v_{n})  \rho(u_{1}\otimes u_{2}\otimes \cdots \otimes u_{n}) \qquad \qquad \qquad \nonumber \\
 + (-1)^{\Delta^{(n)}} \rho(v_{1}\otimes v_{2}\otimes \cdots \otimes v_{n})  \rho(u_{1}\otimes u_{2}\otimes \cdots \otimes u_{n}) \nonumber \\
  \qquad \qquad \qquad =   g(v_{1}\otimes v_{2}\otimes \cdots \otimes v_{n}, u_{1}\otimes u_{2}\otimes \cdots \otimes u_{n}),
\end{eqnarray}
where $\Delta^{(n)}$ is define using a vector basis
$\left\lbrace e_{i_{1}}\otimes e_{i_{2}}\otimes \cdots \otimes e_{i_{n}} \right \rbrace $  of $V^{\otimes n}$ with % \cdots \otimes e_{i_{k}} \otimes
\begin{eqnarray}
\rho(e_{i_{1}}\otimes e_{i_{2}}\otimes \cdots \otimes e_{i_{n}})  \rho(e_{j_{1}}\otimes e_{j_{2}}\otimes \cdots \otimes e_{j_{n}})  \qquad \qquad \qquad \nonumber \\
+(-1)^{\Delta^{(n)}}\rho(e_{j_{1}}\otimes e_{j_{2}}\otimes \cdots \otimes e_{j_{n}})  \rho(e_{i_{1}}\otimes e_{i_{2}}\otimes \cdots \otimes e_{i_{n}}),
\end{eqnarray}
so $\Delta^{(n)}=0$ if $n$ is even and $i_{k}=j_{k}$  an odd number of times or if $n$ is odd and $i_{k}=j_{k}$ an even number of times or still if $i_{k}=j_{k}$ for all $k$ in the product. Otherwise $\Delta^{(n)}=1$ and $g(,)$ is a symmetric bilinear form.
\end{defn}

Let us define a product algebra $Cl^{\otimes n}$ through quotient algebra product
\begin{eqnarray}
\frac{T(V^{\otimes n})}{I_{C}},
\end{eqnarray}
denoting
%\begin{eqnarray}
$\displaystyle \  T(V^{\otimes n})=\bigoplus_{q=0}^{\infty}T_{q}(V^{\otimes n})   $
%\end{eqnarray}
the algebra of contravariant tensors with $ T_{q}(V^{\otimes n}) = (V^{\otimes n})^{\otimes q} $ and $I_{C}$ is a two-sided ideal generated by elements
\begin{eqnarray}
(v_{1}\otimes v_{2}\otimes \cdots \otimes v_{n}) \otimes (u_{1}\otimes u_{2}\otimes \cdots
\otimes u_{n})+ \nonumber \\ (-1)^{\Delta^{(n)}} (v_{1}\otimes v_{2}\otimes \cdots \otimes v_{n}) \otimes (u_{1}\otimes u_{2}\otimes \cdots \otimes u_{n}) \nonumber \\
- g(v_{1}\otimes v_{2}\otimes \cdots \otimes v_{n}, u_{1}\otimes u_{2}\otimes \cdots \otimes u_{n}) 1 ,
\end{eqnarray}
using $\Delta^{(n)}$ previously defined.

\begin{theorem} There is an homomorphism
\begin{eqnarray}
 \Upsilon \left( (Cl^{\otimes n}) \equiv \frac{T(V^{\otimes n} ) }{I_{C}} \right) & \rightarrow C,
\end{eqnarray}
such that $\rho= \Upsilon \circ \gamma $, where $\gamma: V^{\otimes n} \rightarrow Cl^{\otimes n} $.

\end{theorem}

\begin{proof}
Extending the application $\rho$ to $\rho^{'}: T(V^{\otimes n})\rightarrow Cl^{\otimes n}$
\begin{eqnarray}
\hspace{-0.6cm}\rho^{'}(v_{1}^{(1)}\otimes...\otimes v_{1}^{(n)}\otimes...v_{k}^{(n)}\otimes ...\otimes v_{k}^{(n)})&=&\rho^{'}(v_{1}^{(1)}\otimes...\otimes v_{1}^{(n)})...\rho^{'}(v_{k}^{(1)}\otimes...\otimes v_{k}^{(n)}) \nonumber \\
&=&\rho(v_{1}^{(1)}\otimes...\otimes v_{1}^{(n)})...\rho(v_{k}^{(1)}\otimes...\otimes v_{k}^{(n)}) \nonumber.
\end{eqnarray}
Consider now the quotient space $\displaystyle \frac{T(V^{\otimes n})}{Ker (\rho^{'}) }$ and we will do $[\varsigma]=\pi(\varsigma)$, in which $[\varsigma]$, are the equivalence classes of quotient space $\varsigma \in T(V^{\otimes n})$ and $\pi:T(V^{\otimes n})\rightarrow T ( V^{\otimes n})/Ker(\rho^{'} ) $. Then, we can to define the application $\Upsilon: T ( V^{\otimes n})/Ker(\rho^{'}) \rightarrow Cl^{ \otimes n}$ given by  $\Upsilon([\varsigma])=\rho^{'}(\varsigma)$, $\forall \varsigma \in T(V^{\otimes n})$. Note that we have a homomorphism
\begin{eqnarray}
\Upsilon([\varsigma_{1}][\varsigma_{2}])&=&\Upsilon([\varsigma_{1} \otimes \varsigma_{2}]) \nonumber \\
&=&\rho^{'}(\varsigma_{1} \otimes \varsigma_{2}) \nonumber \\
&=&\Upsilon(\varsigma_{1})\Upsilon(\varsigma_{2}). \nonumber
\end{eqnarray}

In view of definition of applications $\rho$ and $\rho^{'}$, we have
\begin{eqnarray}
\rho^{'} \left( (v_{1}^{(1)}\otimes v_{1}^{(2)}\otimes...\otimes v_{1}^{(n)})(v_{2}^{(1)}\otimes v_{2}^{(2)}\otimes...\otimes v_{2}^{(n)}) \qquad \qquad   \right.  & & \nonumber \\
+(-1)^{\Delta^{(n)}}(v_{2}^{(1)}\otimes v_{2}^{(2)}\otimes...\otimes v_{2}^{(n)})(v_{1}^{(1)}\otimes v_{1}^{(2)}\otimes...\otimes v_{1}^{(n)})  & &  \nonumber \\
\left. -g(v_{1}^{(1)}\otimes v_{1}^{(2)}\otimes...\otimes v_{1}^{(n)},v_{2}^{(1)}\otimes v_{2}^{(2)}\otimes...\otimes v_{2}^{(n)}) \right)  & = & 0
\end{eqnarray}
Therefore $I_{C}\subseteq Ker(\rho^{'})$ and consequently $Cl^{\otimes n}\equiv T(V^{\otimes n})/I_{C} \subseteq T(V^{\otimes n})/Ker(\rho^{'}) $. If we restrict homomorphism $\Upsilon$ as $\Upsilon:Cl^{\otimes n} \rightarrow C $, we have that $\rho = \Upsilon \circ \gamma$, keeping in mind that $\Upsilon([\varsigma])=\Upsilon(\gamma(\varsigma))=\rho^{'}(\varsigma)=\rho(\varsigma)$ and the proof is finished.
\end{proof}

We can summarize this result in the following commutative diagram

\begin{center}
\begin{minipage}[t]{6cm}
 \xymatrix{
V^{\otimes n} \ar[r]^-{\gamma}  \ar[dr]_{\rho}&
Cl^{\otimes n}\equiv \frac{T(V^{\otimes n})}{I_{C}} \ar[d]^{\Upsilon} \\
& C
}
\end{minipage}
\end{center}

The previous theorem is related to tensor product of algebras wherein the generators may commute or anticommute among themselves. Particularly, we will use Clifford algebra $(Cl_{1,3}^{+})^{\otimes n}\equiv Cl_{1,3}^{+} \otimes  Cl_{1,3}^{+} \otimes \cdot \cdot \cdot  \otimes  Cl_{1,3}^{+} $, where $n$ is the number of factors.  Now we consider a Clifford algebra $Cl_{1,3}^{+}\simeq Cl_{3,0}$. Let $\Psi^{Cl_{1,3}^{+}}$ be an element of a minimal left ideal $Cl_{1,3}^{+}P$ and $\Psi^{\star}$ an element of minimal right ideal $P \widetilde{ Cl_{1,3}^{+}}$, where $P$ is a primitive idempotent. One define the operators $O^{Cl_{1,3}^{+}}:Cl_{1,3}^{+}P \rightarrow Cl_{1,3}^{+}P$. It is easy to verify that $\langle (\Psi^{Cl_{1,3}^{+}})^{\star}\Psi^{Cl_{1,3}^{+}} \rangle_{0}$ is an inner product. In fact,
\begin{eqnarray}
\langle (\Psi^{Cl_{1,3}^{+}})^{\star}\Psi^{Cl_{1,3}^{+}} \rangle_{0}\geq 0
\end{eqnarray}
and $\langle (\Psi^{Cl_{1,3}^{+}})^{\star}\Psi^{Cl_{1,3}^{+}} \rangle_{0}= 0$ if $\Psi^{Cl_{1,3}^{+}}=0$. Besides $\langle (\Psi^{Cl_{1,3}^{+}})^{\star}\Phi^{Cl_{1,3}^{+}} \rangle_{0}=\langle (\Phi^{Cl_{1,3}^{+}})^{\star}\Psi^{Cl_{1,3}^{+}} \rangle_{0}$ and linear property is satisfied.
\begin{theorem}
The operators $O^{Cl_{1,3}^{+}}$ form a $c^{\star}$-algebra.
\end{theorem}
\begin{proof}
We may define the norm operator
\begin{eqnarray}
\hspace{-1.5cm}\parallel O^{Cl_{1,3}^{+}} \parallel _{\Psi^{Cl_{1,3}^{+}}}&=&\sup_{\Psi_{Cl_{1,3}^{+}}} \left\lbrace \left(\langle (O^{Cl_{1,3}^{+}} \Psi^{Cl_{1,3}^{+}})^{\star}(O^{Cl_{1,3}^{+}}\Psi^{Cl_{1,3}^{+}}) \rangle_{0}\right)^{1/2}, \right. % \nonumber \\ & &
\left. (\langle (\Psi^{Cl_{1,3}^{+}})^{\star}\Psi^{Cl_{1,3}^{+}} \rangle_{0})^{1/2} \ \leq \ 1 \right\rbrace  \nonumber \\
\end{eqnarray}
We have that for $O_{1+2}^{Cl_{1,3}^{+}} \equiv O_{1}^{Cl_{1,3}^{+}}+O_{2}^{Cl_{1,3}^{+}}$

\begin{eqnarray}
\hspace{-1.8cm}\parallel (O_{1+2}^{Cl_{1,3}^{+}}) \parallel _{\Psi^{Cl_{1,3}^{+}}}&=&\sup_{\Psi_{Cl_{1,3}^{+}}} \left\lbrace \left(\langle (O_{1+2}^{Cl_{1,3}^{+}} \Psi^{Cl_{1,3}^{+}})^{\star}(O_{1+2}^{Cl_{1,3}^{+}}\Psi^{Cl_{1,3}^{+}}) \rangle_{0}\right)^{1/2}, % \nonumber \\  & &
(\langle (\Psi^{Cl_{1,3}^{+}})^{\star}\Psi^{Cl_{1,3}^{+}} \rangle_{0})^{1/2} \ \leq \ 1 \right\rbrace \nonumber \\
\hspace{-1.8cm}  & \leq &   \sup_{\Psi_{Cl_{1,3}^{+}}} \left\lbrace \left(\langle (O_{1}^{Cl_{1,3}^{+}} \Psi^{Cl_{1,3}^{+}})^{\star}(O_{1}^{Cl_{1,3}^{+}}\Psi^{Cl_{1,3}^{+}}) \rangle_{0}\right)^{1/2} \right. \nonumber \\
\hspace{-1.8cm} & & \left. + \left(\langle (O_{2}^{Cl_{1,3}^{+}} \Psi^{Cl_{1,3}^{+}})^{\star}(O_{2}^{Cl_{1,3}^{+}}\Psi^{Cl_{1,3}^{+}}) \rangle_{0}\right)^{1/2}, %\nonumber \\ & &
(\langle (\Psi^{Cl_{1,3}^{+}})^{\star}\Psi^{Cl_{1,3}^{+}} \rangle_{0})^{1/2} \ \leq \ 1 \right\rbrace
\nonumber \\
\hspace{-1.8cm} &\leq & \sup_{\Psi_{Cl_{1,3}^{+}}} \left\lbrace \left(\langle (O_{1}^{Cl_{1,3}^{+}} \Psi^{Cl_{1,3}^{+}})^{\star}(O_{1}^{Cl_{1,3}^{+}}\Psi^{Cl_{1,3}^{+}}) \rangle_{0}\right)^{1/2} ,
(\langle (\Psi^{Cl_{1,3}^{+}})^{\star}\Psi^{Cl_{1,3}^{+}} \rangle_{0})^{1/2} \ \leq \  1 \right\rbrace \nonumber \\
%\end{eqnarray}
%\begin{eqnarray}
\hspace{-1.8cm} & & + \sup_{\Psi_{Cl_{1,3}^{+}}} \left\lbrace \left(\langle (O_{2}^{Cl_{1,3}^{+}} \Psi^{Cl_{1,3}^{+}})^{\star}(O_{2}^{Cl_{1,3}^{+}}\Psi^{Cl_{1,3}^{+}}) \rangle_{0}\right)^{1/2},
(\langle (\Psi^{Cl_{1,3}^{+}})^{\star}\Psi^{Cl_{1,3}^{+}} \rangle_{0})^{1/2} \ \leq \ 1 \right\rbrace \nonumber \\
\hspace{-1.8cm} &=&\parallel (O_{1}^{Cl_{1,3}^{+}}) \parallel _{\Psi^{Cl_{1,3}^{+}}} + \parallel (O_{2}^{Cl_{1,3}^{+}}) \parallel _{\Psi^{Cl_{1,3}^{+}}}.
\end{eqnarray}

The other properties that define the norm can be easily verified. We can see that $O^{\star Cl_{1,3}^{+}}=\widetilde{O^{Cl_{1,3}^{+}}}$ %also have the operator $O^{\star Cl_{1,3}^{+}}$
that satisfies the condition
\begin{eqnarray}
\langle (O^{ Cl_{1,3}^{+}}\Phi^{ Cl_{1,3}^{+}})^{\star}\Psi^{ Cl_{1,3}^{+}} \rangle_{0}= \langle (\Phi^{ Cl_{1,3}^{+}})^{\star} O^{\star Cl_{1,3}^{+}}\Psi^{ Cl_{1,3}^{+}} \rangle_{0}.
\end{eqnarray}
 Notice the reversion is an involution that satisfies the axioms of an $c^{\star}$-algebra. In fact,
\begin{eqnarray}
((O^{Cl_{1,3}^{+}})^{\star})^{\star} & = & O^{Cl_{1,3}^{+}},  \nonumber \\
(O_{1}^{Cl_{1,3}^{+}}+O_{2}^{Cl_{1,3}^{+}})^{\star} & = & (O_{1}^{Cl_{1,3}^{+}})^{\star}+(O_{2}^{Cl_{1,3}^{+}})^{\star} \nonumber \\
(O_{1}^{Cl_{1,3}^{+}}O_{2}^{Cl_{1,3}^{+}})^{\star} & = & (O_{2}^{Cl_{1,3}^{+}})^{\star}(O_{1}^{Cl_{1,3}^{+}})^{\star}, \nonumber \\
\parallel O^{Cl_{1,3}^{+}} O^{\star Cl_{1,3}^{+}} \parallel & = & \parallel O^{Cl_{1,3}^{+}} \parallel ^{2} \nonumber \\
(\lambda O^{Cl_{1,3}^{+}})^{\star} & = & \lambda(O^{Cl_{1,3}^{+}})
\end{eqnarray}
with $\lambda \in \mathbb{R}$.  %; the algebra is over the field of real numbers
Therefore we have an $c^{\star}$-algebra of operators on the space of left minimal ideals generated by $\langle \gamma_{3}\gamma_{0}P, \gamma_{1}\gamma_{0}P, \iota\gamma_{3}\gamma_{0}P, \iota\gamma_{1}\gamma_{0}P  \rangle$
\end{proof}

We can extend these results to the tensor product of algebras $Cl_{1,3}^{+}$.

\section{Algebraic formulation} \label{sec3}

In order to establish a relationship with the qubits, we will explore the isomorphism $Cl_{3,0}\simeq Cl^{+}_{1,3}$, $\zeta: Cl_{3,0} \rightarrow Cl^{+}(1,3)$, given by $\zeta(\sigma_{\mu})=\gamma_{\mu}\gamma_{0}$, on what $\sigma_{\mu}$ are generators of Clifford algebra $Cl_{3,0}$.
Then consider the primitive idempotents
\begin{eqnarray}
E & = & \frac{1}{2}(1+\sigma_{3}) \ \in Cl_{3,0} \nonumber \\
P & = & \frac{1}{2}(1+\gamma_{3}\gamma_{0}) \ \in Cl^{+}_{1,3} .
\end{eqnarray}
Qubits can be identified in these algebras as
\begin{equation}
 \begin{array}{rclcl}
|0 \rangle   & \leftrightarrow & \sigma_{3}E   &  \leftrightarrow  &  \gamma_{3}\gamma_{0}P      \\
i|0 \rangle  & \leftrightarrow & \sigma_{1}\sigma_{2} E   &  \leftrightarrow  & \iota\gamma_{3} \gamma_{0}   P \\
 |1 \rangle  & \leftrightarrow & \sigma_{1}E    &  \leftrightarrow  & \gamma_{1}\gamma_{0}P \\
 i|1 \rangle & \leftrightarrow & \sigma_{2}\sigma_{3}  E &  \leftrightarrow  & \iota \gamma_{1} \gamma_{0}  P   ,
\end{array} \label{eq-rel}
\end{equation}
considering $\zeta^{'}:C^{2}\rightarrow Cl^{+}_{1,3}P$, given by $\iota = \gamma_{1}\gamma_{0}\gamma_{2}\gamma_{0}\gamma_{3}\gamma_{0}$, resulting in $\iota \gamma_{3}\gamma_{0} = \gamma_{1}\gamma_{0}\gamma_{2}\gamma_{0}$ and $\iota \gamma_{1}\gamma_{0}=\gamma_{2}\gamma_{0}\gamma_{3}\gamma_{0}$.

A general state of a qubit can be written as an algebraic spinor in $Cl_{1,3}^{+}$, i.e., as an element of an ideal left minimal algebra $Cl_{1,3}^{+}P$:
\begin{eqnarray} \label{estadoB}
\Psi &=&(\alpha_{1}\gamma_{3}\gamma_{0}+\alpha_{2}\iota\gamma_{3}\gamma_{0}+\alpha_{3}\gamma_{1}\gamma_{0}+\alpha_{4}\iota\gamma_{1}\gamma_{0})P \nonumber \\
&=&(\alpha_{1}\gamma_{3}\gamma_{0}+\alpha_{2}\gamma_{1}\gamma_{0}\gamma_{2}\gamma_{0}+\alpha_{3}\gamma_{1}\gamma_{0}+\alpha_{4}\gamma_{2}\gamma_{0}\gamma_{3}\gamma_{0})\frac{1}{2}(1+\gamma_{3}\gamma_{0}),
\end{eqnarray}
with $\alpha_{j} \in \mathbb{R}, \ j=\left\lbrace 1,2,3,4 \right\rbrace $, corresponding to state
\begin{eqnarray}
| \Psi \rangle = (\alpha_{1}+i \alpha_{2})|0 \rangle + (\alpha_{3}+i \alpha_{4})|1 \rangle.
\end{eqnarray}

The corresponding bra is given by:
\begin{eqnarray}
\Psi^{\star} &=&P(\alpha_{1}\gamma_{3}\gamma_{0}-\alpha_{2}\iota\gamma_{3}\gamma_{0}+\alpha_{3}\gamma_{1}\gamma_{0}-\alpha_{4}\iota\gamma_{1}\gamma_{0}) \nonumber \\
&=&\frac{1}{2}(1+\gamma_{3}\gamma_{0})(\alpha_{1}\gamma_{3}\gamma_{0}-\alpha_{2}\gamma_{1}\gamma_{0}\gamma_{2}\gamma_{0}+\alpha_{3}\gamma_{1}\gamma_{0}-\alpha_{4}\gamma_{2}\gamma_{0}\gamma_{3}\gamma_{0}),
\end{eqnarray}
with $\alpha_{j} \in \mathbb{R}, \ j= \left\lbrace 1,2,3,4 \right\rbrace $, associated to state
\begin{eqnarray}
\langle \Psi | =  \langle 0 |(\alpha_{1}+i \alpha_{2}) + \langle 1|(\alpha_{3}+i \alpha_{4}).
\end{eqnarray}

For the description of multipartite systems we need to build the product algebra $\left[Cl_{1,3}^{+}\right]^{\otimes n}$ in according Theorem 1. Also easy to verify vector space isomorphism analogously to reference \cite{Vaz}:
\begin{eqnarray}
Cl^{+ \otimes n}_{1,3}(V^{\otimes n},g)\simeq_{V} \bigotimes_{k=1}^{n} \left(\bigwedge(V)\right)
\end{eqnarray}

We can write multipartite q-bits as tensor products of algebraic spinors in $Cl_{1,3}^{+}$, i.e. as elements of left minimal ideals in $\left( Cl_{1,3}^{+}\right)^{\otimes n}$:
\begin{eqnarray}
\Psi &\in& \left[Cl_{1,3}^{+}\right]P \otimes \left[Cl_{1,3}^{+}\right]P \cdots \left[Cl_{1,3}^{+}\right]P \nonumber \\
&=& \left[Cl_{1,3}^{+} \otimes Cl_{1,3}^{+} \otimes \cdots \otimes Cl_{1,3}^{+} \right][P \otimes P \otimes \cdots \otimes P] \nonumber \\
&\equiv & \left[Cl_{1,3}^{+}\right]^{\otimes n}\left[P \right]^{\otimes n}.
\end{eqnarray}
For example, in the Hilbert space, a general bipartite state
\begin{eqnarray}
| \Psi \rangle & = & \left(\alpha_{1}  + i \alpha_{2} \right) |00 \rangle    +
\left(\alpha_{3}  + i \alpha_{4} \right) |01 \rangle  + \left(\alpha_{5}  + i \alpha_{6} \right) |10 \rangle
+ \left(\alpha_{7}  + i \alpha_{8} \right) |11 \rangle \nonumber \\
\end{eqnarray}
corresponds to
\begin{eqnarray} \label{bipartite}
\Psi &=&  [\alpha_{1} (\gamma_{3}\gamma_{0})\otimes (\gamma_{3}\gamma_{0})
+ \alpha_{2}  (\gamma_{3}\gamma_{0} )\otimes (\gamma_{1}\gamma_{0}\gamma_{2}\gamma_{0}) \nonumber \\
& & + \alpha_{3} (\gamma_{3}\gamma_{0})\otimes (\gamma_{1}\gamma_{0})
+ \alpha_{4} (\gamma_{3}\gamma_{0})\otimes (\gamma_{2}\gamma_{0}\gamma_{3}\gamma_{0}) \nonumber \\
& & + \alpha_{5} (\gamma_{1}\gamma_{0})\otimes (\gamma_{3}\gamma_{0})
+ \alpha_{6} (\gamma_{1}\gamma_{0})  \otimes  (\gamma_{1}\gamma_{0}\gamma_{2}\gamma_{0}) \nonumber \\
& & + \alpha_{7} (\gamma_{1}\gamma_{0})\otimes (\gamma_{1}\gamma_{0})
+ \alpha_{8} (\gamma_{1}\gamma_{0})\otimes (\gamma_{2}\gamma_{0}\gamma_{3}\gamma_{0}) ]P^{\otimes 2}.
\end{eqnarray}
with $\alpha_{j} \in \mathbb{R}, \ j= \left\lbrace 1,2,...,8 \right\rbrace $.

An entangled bipartite state is a state that cannot be written as:
\begin{eqnarray}
\Psi= \Psi_{1} \otimes \Psi_{2}.
\end{eqnarray}
where the $1,2$ indices refer to the $1,2 $ subsystems, respectively; otherwise, the state is said to be factorable. We will use bipartite states for quantum teleportation, with maximally entangled Bell-like states. Since we are using real Clifford algebras, quantum logic gates can be algebraically identified by the reversion operation:
\begin{eqnarray}
\widetilde{\sum_{k_{\nu}}A_{[k_{\nu}]}} \equiv \sum_{k_{\nu}}\widetilde{A}_{[k_{\nu}]}=\sum_{k_{\nu}}(-1)^{k_{\nu}(k_{\nu}-1)/2}A_{[k_{\nu}]}
\end{eqnarray}
 if $\Gamma_{[k]} \in \{\gamma_{0}, \ \gamma_{i}\gamma_{j} \ (i \neq j \neq 0), \ \gamma_{0} \gamma_{i}\gamma_{j}  \ (i \neq j \neq 0), \gamma_{0}\gamma_{1} \gamma_{2}\gamma_{3}\}$ and
\begin{eqnarray}
\widetilde{\sum_{k_{\nu}}A_{[k_{\nu}]}} \equiv \sum_{k_{\nu}}\widetilde{A}_{[k_{\nu}]}=\sum_{k_{\nu}}(-1)^{[k_{\nu}(k_{\nu}-1)-2]/2}A_{[k_{\nu}]},
\end{eqnarray}
otherwise. Note that
\begin{eqnarray}
\textrm{\stackon[-8pt]{$ \sum_{k_{\nu}}A_{[k_{\nu}]}\sum_{k_{\nu}}B_{[k_{\nu}]} $}{\vstretch{1.5}{\hstretch{2.9}{\widetilde{\phantom{\qquad \qquad }}}}} }
  & = & \widetilde{\sum_{k_{\nu}}B_{[k_{\nu}]}}\widetilde{\sum_{k_{\nu}}A_{[k_{\nu}]}}
\end{eqnarray}
Generalizing for $n$ subsystems, we have
\begin{eqnarray}
\widetilde{\left[\left(\sum_{k_{\nu_{i}}}A_{[k_{\nu_{i}}]}\right)^{\otimes n}\right]}\equiv \left(\sum_{k_{\nu_{1}}}\widetilde{A}_{[k_{\nu_{1}}]} \right) \otimes \left(\sum_{k_{\nu_{2}}}\widetilde{A}_{[k_{\nu_{2}}]} \right) \otimes \cdots \otimes \left(\sum_{k_{\nu_{n}}}\widetilde{A}_{[k_{\nu_{n}}]} \right).
\end{eqnarray}
The condition analogous to unitarity is given by:
\begin{eqnarray}
\left[\left(\sum_{k_{\nu_{i}}}A_{[k_{\nu_{i}}]}\right)^{\otimes n}\right] \widetilde{\left[\left(\sum_{k_{\nu_{i}}}A_{[k_{\nu_{i}}]}\right)^{\otimes n}\right]}=1^{\otimes n}.
\end{eqnarray}
It is noteworthy that the above definition corresponds to unitary operators that are factorable into unitary operators. In the transposition into Hilbert's space language these facts are evident. We will explore these ideas in later sections.

%\newpage
\subsection{Representation of braid group and quantum teleportation} \label{sec4}
Next we will show how we can construct unitary representations of the braid group using Clifford algebra $Cl_ {1,3}$. Moreover, we will deduce an equation for quantum teleportation in this formalism. To do so, consider the exponential of a multivector $\Gamma \in Cl_ {p, q}$ i.e.
\begin{eqnarray}
\exp \Gamma = \sum_{n=0}^{\infty}\frac{\Gamma^{n}}{n!}.
\end{eqnarray}

We define the elements $B_{1}^{Cl_{1,3}^{+}}$ and $B_{2}^{Cl_{1,3}^{+}}$ as
\begin{eqnarray}
B_{1}^{Cl_{1,3}^{+}} & = & \exp \left(\frac{\pi}{4}\gamma_{1}\gamma_{0}\gamma_{2}\gamma_{0}\right) \\
B_{2}^{Cl_{1,3}^{+}} & = & \exp \left(\frac{\pi}{4}\gamma_{2}\gamma_{0}\gamma_{3}\gamma_{0}\right) ,
\end{eqnarray}
which implies
\begin{eqnarray}
B_{1}^{Cl_{1,3}^{+}} &=&\frac{1}{\sqrt{2}}\left( 1+\gamma_{1}\gamma_{0}\gamma_{2}\gamma_{0} \right) \\
B_{2}^{Cl_{1,3}^{+}} &=&\frac{1}{\sqrt{2}}\left( 1+\gamma_{2}\gamma_{0}\gamma_{3}\gamma_{0} \right),
\end{eqnarray}
since  $(\gamma_{1}\gamma_{0}\gamma_{2}\gamma_{0})^{2}=-1$ and $(\gamma_{2}\gamma_{0}\gamma_{3}\gamma_{0})^{2}=-1$. The elements $B_{1}^{Cl_{1,3}^{+}}$ and $B_{2}^{Cl_{1,3}^{+}}$ don't commute but obey the following property:
\begin{eqnarray}
B_{1}^{Cl_{1,3}^{+}}B_{2}^{Cl_{1,3}^{+}}B_{1}^{Cl_{1,3}^{+}}=B_{2}^{Cl_{1,3}^{+}}B_{1}^{Cl_{1,3}^{+}}B_{2}^{Cl_{1,3}^{+}}. \label{eq.tranca}
\end{eqnarray}
In fact,

\begin{eqnarray}
B_{1}^{Cl_{1,3}^{+}}B_{2}^{Cl_{1,3}^{+}}B_{1}^{Cl_{1,3}^{+}} & = & \frac{1}{2} \left( 1+\gamma_{1}\gamma_{0}\gamma_{2}\gamma_{0} \right)  \left(1+\gamma_{2}\gamma_{0}\gamma_{3}\gamma_{0} \right) B_{1}^{Cl_{1,3}^{+}} \nonumber \\
& = & \frac{1}{2\sqrt{2}}  \left(1 + \gamma_{1}\gamma_{0}\gamma_{2}\gamma_{0} + \gamma_{2}\gamma_{0}\gamma_{3}\gamma_{0}  + \gamma_{1}\gamma_{0}\gamma_{3}\gamma_{0} \right) \left( 1+\gamma_{1}\gamma_{0}\gamma_{2}\gamma_{0} \right) \nonumber \\
& = &  \frac{1}{\sqrt{2}}  \left(\gamma_{1}\gamma_{0}\gamma_{2}\gamma_{0} + \gamma_{2}\gamma_{0}\gamma_{3}\gamma_{0}  \right) ,
\end{eqnarray}
and
\begin{eqnarray}
B_{2}^{Cl_{1,3}^{+}}B_{1}^{Cl_{1,3}^{+}}B_{2}^{Cl_{1,3}^{+}} & = & \frac{1}{2}\left( 1+\gamma_{2}\gamma_{0}\gamma_{3}\gamma_{0} \right)\left( 1+\gamma_{1}\gamma_{0}\gamma_{2}\gamma_{0} \right) B_{2}^{Cl_{1,3}^{+}} \nonumber  \\
 & = & \frac{1}{2\sqrt{2}}  \left(1 + \gamma_{2}\gamma_{0}\gamma_{3}\gamma_{0} + \gamma_{1}\gamma_{0}\gamma_{2}\gamma_{0}  - \gamma_{1}\gamma_{0}\gamma_{3}\gamma_{0} \right) \left( 1+\gamma_{2}\gamma_{0}\gamma_{3}\gamma_{0} \right) \nonumber \\
& = & \frac{1}{\sqrt{2}}  \left(\gamma_{1}\gamma_{0}\gamma_{2}\gamma_{0} + \gamma_{2}\gamma_{0}\gamma_{3}\gamma_{0}  \right) .
\end{eqnarray}

Therefore $B_{1}^{Cl_{1,3}^{+}}$ and $B_{2}^{Cl_{1,3}^{+}}$, satisfying equation (\ref{eq.tranca}). A representation of these elements in terms of Pauli's matrices and satisfying equation (\ref{eq.tranca}), can be associated with generators of a braid group, first devised in 1925 by Emil Artin \cite{artin}. Considering the Pauli's matrices
\begin{eqnarray}
\sigma _{1}=\left(
\begin{array}{cc}
0 & 1 \\
1 & 0%
\end{array}%
\right), \ \
\sigma _{2}=\left(
\begin{array}{cc}
0 & -i \\
i & 0%
\end{array}%
\right), \ \
\sigma _{3}=\left(
\begin{array}{cc}
1 & 0 \\
0 & -1%
\end{array}%
\right),
\end{eqnarray}
we can write operators $B_{1}^{Cl_{1,3}^{+}}$ and $B_{2}^{Cl_{1,3}^{+}}$ as
\begin{eqnarray}
B _{1} & = & \frac{1}{\sqrt{2}} \left( I +i \sigma_{3}\right) \nonumber \\
       & = & \frac{1}{\sqrt{2}}\left( \begin{array}{cc}
1+i & 0 \\
0 & 1-i%
\end{array}%
\right), \end{eqnarray}
\begin{eqnarray}
B _{2} & = & \frac{1}{\sqrt{2}} \left( 1+ i\sigma_{1}  \right) \nonumber \\
       & = & \frac{1}{\sqrt{2}}\left( \begin{array}{cc}
1 & i \\
i & 1
\end{array} %
\right).
\end{eqnarray}
with $det B_{1} = 1$ and $det B_{2} = 1$ . The action of these operators on qubits in the usual formulation is given by:
\begin{eqnarray}
B_{1}| 0 \rangle &=& \frac{(1+i)}{\sqrt{2}} | 0 \rangle \nonumber \\
B_{1}| 1 \rangle &=& \frac{(1-i)}{\sqrt{2}} | 1 \rangle \nonumber \\
B_{2}| 0 \rangle &=& \frac{1}{\sqrt{2}} (| 0 \rangle + i| 1 \rangle) \nonumber \\
B_{2}| 1 \rangle &=& \frac{1}{\sqrt{2}} (i| 0 \rangle + | 1 \rangle). \label{eq-b1b2}
\end{eqnarray}

Now, from equation (\ref{eq-rel}) we have
\begin{eqnarray}
 \gamma_{1}\gamma_{0} \gamma_{2}\gamma_{0} & \leftrightarrow &  \sigma_{1}\sigma_{2} = i\sigma_{3} \nonumber \\
 \gamma_{2}\gamma_{0} \gamma_{3}\gamma_{0} & \leftrightarrow &  \sigma_{2}\sigma_{3} = i\sigma_{1} .
\end{eqnarray}
So, within our construction, we can analyze the action of these operators in terms of algebraic spinors, that it,
\begin{eqnarray}
B_{1}^{Cl_{1,3}^{+}} \gamma_{3}\gamma_{0}P &=& \frac{1}{\sqrt{2}} \gamma_{3}\gamma_{0}P  + \frac{1}{\sqrt{2}} \gamma_{1}\gamma_{0}\gamma_{2}\gamma_{0}\gamma_{3}\gamma_{0} P \nonumber \\
& = & \frac{1}{\sqrt{2}} \left(\gamma_{3}\gamma_{0} + \gamma_{1}\gamma_{0}\gamma_{2}\gamma_{0} \right)P, \label{eq-b10}
\\
B_{1}^{Cl_{1,3}^{+}}\gamma_{1}\gamma_{0}P & = & \frac{1}{\sqrt{2}} \gamma_{1}\gamma_{0}P +
\frac{1}{\sqrt{2}} \gamma_{1}\gamma_{0}\gamma_{2}\gamma_{0}\gamma_{1}\gamma_{0} P \nonumber \\
&=& \frac{1}{\sqrt{2}} (\gamma_{1}\gamma_{0}-\gamma_{2}\gamma_{0}\gamma_{3}\gamma_{0})P , \label{eq-b11}
\\
B_{2}^{Cl_{1,3}^{+}}\gamma_{3}\gamma_{0}P & = &  \frac{1}{\sqrt{2}} \gamma_{3}\gamma_{0}P  + \frac{1}{\sqrt{2}} \gamma_{2}\gamma_{0}\gamma_{3}\gamma_{0}\gamma_{3}\gamma_{0} P \nonumber \\
&=& \frac{1}{\sqrt{2}} (\gamma_{3}\gamma_{0}+\gamma_{2}\gamma_{0}\gamma_{3}\gamma_{0})P,
\\
B_{2}^{Cl_{1,3}^{+}}\gamma_{1}\gamma_{0}P & = & \frac{1}{\sqrt{2}} \gamma_{1}\gamma_{0}P +
\frac{1}{\sqrt{2}} \gamma_{2}\gamma_{0}\gamma_{3}\gamma_{0}\gamma_{1}\gamma_{0} P \nonumber \\
&=& \frac{1}{\sqrt{2}} (\gamma_{1}\gamma_{0}\gamma_{2}\gamma_{0}+\gamma_{1}\gamma_{0})P,
\end{eqnarray}
getting algebric relations equivalent to equations (\ref{eq-b1b2}), where we use the anticommutation relations of the elements $\displaystyle \gamma_{i}$, where $\left(\gamma_{3}\gamma_{0}\right)^{2}=1$ and  $\displaystyle \gamma_{3}\gamma_{0}P = \gamma_{3}\gamma_{0} \frac{1}{2}(1+\gamma_{3}\gamma_{0}) = P$.

The description for $n$ subsystems is given from the tensor product of the operators.
\begin{eqnarray}
B_{i}^{Cl_{1,3}^{+}} \otimes B_{i}^{Cl_{1,3}^{+}} \otimes \cdots \otimes B_{i}^{Cl_{1,3}^{+}} & \equiv & \left(B_{i}^{Cl_{1,3}^{+}}\right)^{\otimes n}, \ \ \ i=1,2
\end{eqnarray}
which are also braid group generators:
\begin{eqnarray}
\left(B_{1}^{Cl_{1,3}^{+}}\right)^{\otimes n} \left(B_{2}^{Cl_{1,3}^{+}}\right)^{\otimes n}\left(B_{1}^{Cl_{1,3}^{+}}\right)^{\otimes n}=\left(B_{2}^{Cl_{1,3}^{+}}\right)^{\otimes n} \left(B_{1}^{Cl_{1,3}^{+}}\right)^{\otimes n}\left(B_{2}^{Cl_{1,3}^{+}}\right)^{\otimes n}.
\end{eqnarray}

Note that the norm of the elements $B_{1}^{Cl_{1,3}^{+}}$ and $B_{2}^{Cl_{1,3}^{+}}$ is defined from the reversion operation
\begin{eqnarray}
\left| B_{1}^{Cl_{1,3}^{+}} \right|^{2} & = & \langle \widetilde{B_{1}^{Cl_{1,3}^{+}}} B_{1}^{Cl_{1,3}^{+}}  \rangle_{0} \nonumber \\
& = & \frac{1}{2} \langle (1-\gamma_{1}\gamma_{0}\gamma_{2}\gamma_{0}) (1+\gamma_{1}\gamma_{0}\gamma_{2}\gamma_{0}) \rangle_{0} \nonumber \\
& = & 1, \\
\left| B_{2}^{Cl_{1,3}^{+}} \right|^{2} & = & \langle \widetilde{B_{2}^{Cl_{1,3}^{+}}} B_{2}^{Cl_{1,3}^{+}}  \rangle_{0} \nonumber \\
& = & \frac{1}{2} \langle (1-\gamma_{2}\gamma_{0}\gamma_{3}\gamma_{0}) (1+\gamma_{2}\gamma_{0}\gamma_{3}\gamma_{0}) \rangle_{0} \nonumber \\
& = & 1,
\end{eqnarray}
which implies by eq. (\ref{eq4}) that $B_{1}^{Cl_{1,3}^{+}} \in spin_{+}(1,3)$ and $B_{2}^{Cl_{1,3}^{+}} \in spin_{+}(1,3)$. Consequently $\langle B_{1}^{Cl_{1,3}^{+}}, B_{2}^{Cl_{1,3}^{+}} \rangle \leq spin_{+}(1,3)$, i.e. , $\langle B_{1}^{Cl_{1,3}^{+}}, B_{2}^{Cl_{1,3}^{+}} \rangle $ is a finite subgroup of $spin_{+}(1,3)$. For $n$ subsystems %o resultado
\begin{eqnarray}
\left( B_{i}^{Cl_{1,3}^{+}} \right)^{\otimes n} \in spin_{+}(1,3) \times spin_{+}(1,3) \times \cdots \times spin_{+}(1,3)
\end{eqnarray}
and associated Lie algebra
\begin{eqnarray}
 Lie \left[spin_{+}(1,3)\right]^{\otimes n} = \oplus_{i=1}^{n}Lie [spin_{+}(1,3)].
\end{eqnarray}

Following our proposed algebraic construction, the Bell states \cite{bell64}
\begin{eqnarray}
|\Phi^{+} \rangle & = & \frac{1}{\sqrt{2}} \left( |0\rangle \otimes |0\rangle + |1\rangle \otimes |1\rangle   \right) \nonumber \\
|\Phi^{-} \rangle & = & \frac{1}{\sqrt{2}} \left( |0\rangle \otimes |0\rangle - |1\rangle \otimes |1\rangle   \right) \nonumber \\
 |\Psi^{+} \rangle & = & \frac{1}{\sqrt{2}} \left( |0\rangle \otimes |1\rangle + |1\rangle \otimes |0\rangle   \right) \nonumber \\
|\Psi^{-} \rangle & = & \frac{1}{\sqrt{2}} \left( |0\rangle \otimes |1\rangle - |1\rangle \otimes |0\rangle   \right), \label{eq-bell}
\end{eqnarray}
are written algebraically as
\begin{eqnarray}
\Phi^{+Cl_{1,3}}&=&\frac{1}{\sqrt{2}}(\gamma_{3}\gamma_{0}\otimes \gamma_{3}\gamma_{0}+\gamma_{1}\gamma_{0}\otimes \gamma_{1}\gamma_{0})P^{\otimes 2} \nonumber \\
\Phi^{-Cl_{1,3}}&=&\frac{1}{\sqrt{2}}(\gamma_{3}\gamma_{0}\otimes \gamma_{3}\gamma_{0}-\gamma_{1}\gamma_{0}\otimes \gamma_{1}\gamma_{0})P^{\otimes 2} \nonumber \\
\Psi^{+Cl_{1,3}}&=&\frac{1}{\sqrt{2}}(\gamma_{3}\gamma_{0}\otimes \gamma_{1}\gamma_{0}+\gamma_{1}\gamma_{0}\otimes \gamma_{3}\gamma_{0})P^{\otimes 2} \nonumber \\
\Psi^{-Cl_{1,3}}&=&\frac{1}{\sqrt{2}}(\gamma_{3}\gamma_{0}\otimes \gamma_{1}\gamma_{0}-\gamma_{1}\gamma_{0}\otimes \gamma_{3}\gamma_{0})P^{\otimes 2}.
\end{eqnarray}

Analyzing operator action $(B_{1}^{Cl_{1,3}^{+} })^{\otimes 2}$ on the Bell states $\Psi^{+Cl_{1,3}}$ and $\Psi^{-Cl_{1,3}}$, in the algebraic formulation, we have
\begin{eqnarray}
(B_{1}^{Cl_{1,3}^{+} })^{\otimes 2} \Psi^{\pm Cl_{1,3}}
& = & \frac{1}{\sqrt{2}} (B_{1}^{Cl_{1,3}^{+} })^{\otimes 2} \gamma_{3}\gamma_{0} P\otimes \gamma_{1}\gamma_{0} P \pm
\frac{1}{\sqrt{2}} (B_{1}^{Cl_{1,3}^{+} })^{\otimes 2}  \gamma_{1}\gamma_{0} P \otimes \gamma_{3}\gamma_{0} P \nonumber \\
   & = & \frac{1}{2\sqrt{2}}\left[ \left( \gamma_{3}\gamma_{0} + \iota \gamma_{3}\gamma_{0} \right)P\otimes\left( \gamma_{1}\gamma_{0} - \iota \gamma_{1}\gamma_{0} \right)P  \right]  \nonumber \\
   &  & \pm \frac{1}{2\sqrt{2}}\left[ \left( \gamma_{1}\gamma_{0} - \iota \gamma_{1}\gamma_{0} \right)P  \right] \otimes
    \left( \gamma_{3}\gamma_{0} + \iota \gamma_{3}\gamma_{0} \right)P \nonumber \\
   & = &  \frac{1}{2\sqrt{2}} \left[\left( 1+ \iota \right)\otimes \left( 1 - \iota \right) \right] \left[ \gamma_{3}\gamma_{0}P \otimes \gamma_{1}\gamma_{0} P\right]    \nonumber \\
&  &  \pm \frac{1}{2\sqrt{2}} \left[\left( 1+ \iota \right)\otimes \left( 1 - \iota \right) \right] \left[ \gamma_{1}\gamma_{0}P \otimes \gamma_{3}\gamma_{0} P\right]  \nonumber \\
& = &
\frac{1}{2} \left[\left( 1+ \iota \right)\otimes \left( 1 - \iota \right) \right] \frac{1}{\sqrt{2}}
\left[ \gamma_{3}\gamma_{0} \otimes \gamma_{1}\gamma_{0} P^{\otimes 2}   \pm  \gamma_{1}\gamma_{0} \otimes \gamma_{3}\gamma_{0} P^{\otimes 2}\right]  \nonumber \\
& = & \Psi^{\pm Cl_{1,3}}, \label{eq63}
\end{eqnarray}
where we have used equations (\ref{eq-b10}) and (\ref{eq-b11}), and the following equivalence for the representations of the states $i|01\rangle$ and $i|10\rangle$,
\begin{eqnarray}
 i|01\rangle & \leftrightarrow & \left\lbrace \begin{array}{l}
\iota \gamma_{3} \gamma_{0} \otimes \gamma_{1}\gamma_{0} P^{\otimes 2} \ = \ \gamma_{1}\gamma_{0}\gamma_{2}\gamma_{0}P \otimes \gamma_{1}\gamma_{0} P \\
\gamma_{3} \gamma_{0} \otimes \iota  \gamma_{1}\gamma_{0} P^{\otimes 2} \ = \ \gamma_{3}\gamma_{0}P \otimes \gamma_{2}\gamma_{0}\gamma_{3}\gamma_{0} P ,
\end{array}\right.  \\
 i|10\rangle & \leftrightarrow & \left\lbrace \begin{array}{l}
\iota \gamma_{1} \gamma_{0} \otimes \gamma_{3}\gamma_{0} P^{\otimes 2} \ = \ \gamma_{2}\gamma_{0}\gamma_{3}\gamma_{0}P \otimes \gamma_{3}\gamma_{0} P \\
\gamma_{1} \gamma_{0} \otimes \iota  \gamma_{3}\gamma_{0} P^{\otimes 2} \ = \ \gamma_{1}\gamma_{0}P \otimes \gamma_{1}\gamma_{0}\gamma_{2}\gamma_{0} P .
\end{array}\right.
\end{eqnarray}

Eq. (\ref{eq63}) indicates that the operator $[B_{1}^{Cl_{1,3}^{+}  }]^{\otimes 2}$ leaves the states $\Psi^{+Cl_{1,3}^{+}}$ and $\Psi^{-Cl_{1,3}^{+}}$ invariant, however the operator $[B_{2}^{Cl_{1,3}^{+}  }]^{\otimes 2}$ don't have the same property.
Concluding the teleportation equation \cite{Preskill} can be written in terms of the algebraic elements of the $Cl_{1,3}^{+}$ as:
\begin{eqnarray}
\psi^{Cl_{1,3}^{+}}_{C}\Psi^{+Cl_{1,3}^{+}}_{AB}
& = &
\frac{1}{\sqrt{2}}\left[a \gamma_{3}\gamma_{0}+b \gamma_{1}\gamma_{0} \right]\otimes \left[(\gamma_{3} \gamma_{0})\otimes (\gamma_{1} \gamma_{0}) +
(\gamma_{1} \gamma_{0}) \otimes (\gamma_{3} \gamma_{0}) \right]P^{\otimes 3} \nonumber \\
& = & \frac{1}{\sqrt{2}}\left[
  a (\gamma_{3}\gamma_{0}) \otimes (\gamma_{3} \gamma_{0}) \otimes (\gamma_{1} \gamma_{0})
+ a (\gamma_{3}\gamma_{0}) \otimes (\gamma_{1} \gamma_{0}) \otimes (\gamma_{3} \gamma_{0}) \right. \nonumber \\ & & \left.
+ b (\gamma_{1}\gamma_{0}) \otimes (\gamma_{3} \gamma_{0}) \otimes (\gamma_{1} \gamma_{0})
+ b (\gamma_{1}\gamma_{0}) \otimes (\gamma_{1} \gamma_{0}) \otimes (\gamma_{3} \gamma_{0}) \right]P^{\otimes 3}  \nonumber \\
& = & \frac{1}{2} \left[
  a (\Phi^{+}_{CA} + \Phi^{-}_{CA} )  \otimes (\gamma_{1} \gamma_{0})
+ a (\Psi^{+}_{CA} + \Psi^{-}_{CA} )  \otimes (\gamma_{3} \gamma_{0})
\right. \nonumber \\ & & \left.
+ b (\Psi^{+}_{CA} - \Psi^{-}_{CA} )  \otimes (\gamma_{1} \gamma_{0})
+ b (\Phi^{+}_{CA} - \Phi^{-}_{CA} )  \otimes (\gamma_{3} \gamma_{0})  \right] P^{\otimes 3} \nonumber \\
& = & \frac{1}{2} \left[
  \Phi^{+}_{CA} \otimes (a\gamma_{1} \gamma_{0} + b\gamma_{3} \gamma_{0}) +
  \Phi^{-}_{CA} \otimes (a\gamma_{1} \gamma_{0} - b\gamma_{3} \gamma_{0})
\right. \nonumber \\ & & \left.
+  \Psi^{+}_{CA}   \otimes (a\gamma_{3} \gamma_{0} + b\gamma_{1} \gamma_{0}  )
+  \Psi^{-}_{CA}   \otimes (a\gamma_{3} \gamma_{0} - b\gamma_{1} \gamma_{0} )  \right] P^{\otimes 3} \nonumber \\
& = & \frac{1}{2} \left[
   \Phi^{+Cl_{1,3}^{+}}_{CA} \otimes (\gamma_{1}\gamma_{0}) \psi^{Cl_{1,3}^{+}}_{B} +
   \Phi^{-Cl_{1,3}^{+}}_{CA} \otimes (\gamma_{1} \gamma_{0}\gamma_{3} \gamma_{0}) \psi^{Cl_{1,3}^{+}}_{B} \right. \nonumber \\  & & \left.
 + \Psi^{+Cl_{1,3}^{+}}_{CA} \otimes \psi^{Cl_{1,3}^{+}}_{B} % \nonumber \\  & &
 + \Psi^{-Cl_{1,3}^{+}}_{CA} \otimes (\gamma_{3}\gamma_{0}) \psi^{Cl_{1,3}^{+}}_{B} \right] P^{\otimes 3} ,
\end{eqnarray}
 where $|\psi \rangle \leftrightarrow \psi^{Cl_{1,3}^{+}}_{C} = \left( a \gamma_{3}\gamma_{0} + b \gamma_{1}\gamma_{0} \right) P $, with $a,b \in \mathbb{R}$, is the state of $C$ subsystem to be teleported.  In addition we use the algebraic representation of the states in a Bell's base, where $\Psi^{+Cl_{1,3}^{+}}_{CA}$,  $\Psi^{-Cl_{1,3}^{+}}_{CA}$, $\Phi^{+Cl_{1,3}^{+}}_{CA}$, $\Phi^{+Cl_{1,3}^{+}}_{CA} $ and $\Psi^{+Cl_{1,3}^{+}}_{AB}$ are the Bell states on the subsystems $CA$ and $AB$.
%\begin{eqnarray}
% |00\rangle & \leftrightarrow & \gamma_{3}\gamma_{0}\gamma_{3}\gamma_{0} = \frac{\sqrt{2}}{2} \left[  \right] \nonumber \\
%\end{eqnarray}

% Podemos também verificar que:
%\begin{eqnarray}
%\widetilde{B_{1}}^{Cl_{1,3}^{+}}&=&\frac{1}{\sqrt{2}}(1-\gamma_{1}\gamma_{0}\gamma_{2}\gamma_{0}) \nonumber \\
%\widetilde{B_{2}}^{Cl_{1,3}^{+}}&=&\frac{1}{\sqrt{2}}(1-\gamma_{2}\gamma_{0}\gamma_{3}\gamma_{0}) \nonumber \\
%\end{eqnarray}

These operators appear in the context of topological quantum computation. Several experimental achievements are associated with Majorana fermions. However, the construction outlined herein does not incorporate such particles. In the next section, we will present a prescription for the description of Majorana fermions through Clifford algebra $Cl^{+}_{1,3}$.

%\newpage

\subsection{Majorana operator and supersymmetry} \label{sec5}
Majorana's fermions have aroused much interest in potential applications in topological quantum computing. In these systems there is a ground state degeneracy that allows natural protection against decoherence. Our goal is to show that we can describe a system of three Majorana fermions using Clifford algebra $ Cl_ {1,3} $. Our formulation follows the prescription set out in the references \cite{Lee, Yu, livrowilczek}.

Majorana operators satisfy the Clifford algebra
\begin{eqnarray}
\{\Gamma_{i}, \Gamma_{j}\}=2 \delta_{i,j}, \ \ \ \textrm{ with } i=1,2,3 \ ,
\end{eqnarray}
where a realization of these operators, within our construction, is given by
\begin{eqnarray}
\Gamma_{1} & \equiv & \gamma_{1}\gamma_{0} \otimes 1 \nonumber \\
\Gamma_{2} & \equiv & \gamma_{3}\gamma_{0} \otimes 1 \nonumber \\
\Gamma_{3} & \equiv & \gamma_{2}\gamma_{0} \otimes \gamma_{1}\gamma_{0}.
\end{eqnarray}

We can build a Hamiltonian through three Majorana fermions
\begin{eqnarray}
H^{Cl_{1,3}^{+},M}=-(\iota \otimes 1) \left[a(\gamma_{1}\gamma_{0} \otimes 1)+b(\gamma_{3}\gamma_{0} \otimes 1)+c(\gamma_{2}\gamma_{0} \otimes \gamma_{1}\gamma_{0}) \right] ,
\end{eqnarray}
with the operator $P^{Cl_{1,3}^{+},M}$, that implements the electron number parity,
\begin{eqnarray}
P^{Cl_{1,3}^{+},M}=\gamma_{2}\gamma_{0} \otimes \gamma_{3}\gamma_{0} \label{eq-paridade}
\end{eqnarray}
and the emergent Majorana operator,
\begin{eqnarray}
\Gamma^{e;Cl_{1,3}^{+},M}=-(\iota \otimes 1)(\gamma_{1}\gamma_{0} \otimes 1)(\gamma_{3}\gamma_{0} \otimes 1)(\gamma_{2}\gamma_{0} \otimes \gamma_{2}\gamma_{0}) \label{eq-emergente}
\end{eqnarray}
establish a complete algebra for the description of Majorana's fermions.

By using operators (\ref{eq-paridade}) and (\ref{eq-emergente}),  Lee and Wilczek \cite{Lee} showed that there is a doubling of degeneracy at any energy level and not just in the ground state. This is a consequence of the following relationships:
\begin{eqnarray}
(\Gamma^{e;Cl_{1,3}^{+},M})^{2} & = & 1 \\
\left[ \Gamma^{e;Cl_{1,3}^{+},M}, \Gamma_{j}^{e;Cl_{1,3}^{+},M} \right]  & = &  0 , \ \ \ \textrm{ with } i=1,2,3 \\
\left[ \Gamma^{e;Cl_{1,3}^{+},M}, H^{Cl_{1,3}^{+},M} \right]  & = &  0 \\
\left\lbrace \Gamma^{e;Cl_{1,3}^{+},M},P^{Cl_{1,3}^{+},M} \right\rbrace & = & 0
\end{eqnarray}

On the other hand the generators of braid group are given by
\begin{eqnarray}
B_{1}^{Cl_{1,3}^{+},M}&=&\exp[\theta (\gamma_{1}\gamma_{0} \gamma_{3}\gamma_{0} \otimes 1)] \nonumber \\
B_{2}^{Cl_{1,3}^{+},M}&=&\exp[\theta (\gamma_{3}\gamma_{0} \gamma_{2}\gamma_{0} \otimes \gamma_{1}\gamma_{0})]
\end{eqnarray}
satisfying the following relation
\begin{eqnarray}
B_{1}^{Cl_{1,3}^{+},M}B_{2}^{Cl_{1,3}^{+},M}B_{1}^{Cl_{1,3}^{+},M}=B_{2}^{Cl_{1,3}^{+},M}B_{1}^{Cl_{1,3}^{+},M}B_{2}^{Cl_{1,3}^{+},M},
\end{eqnarray}
and $\left[ B_{i}^{Cl_{1,3}^{+},M}, \Gamma^{e;Cl_{1,3}^{+},M} \right] = 0$ for $j=1,2$, which indicates that the emergent Majorana operator, $\Gamma^{e;Cl_{1,3}^{+},M}$, plays a role as a braid group invariant.

%These operators may preserve entanglement. For example,
%\begin{eqnarray}
%B_{1}^{Cl_{1,3}^{+},M}\Psi^{+Cl_{1,3}}=-cos(\theta)\Phi^{-Cl_{1,3}}+sen(\theta)\Psi^{+Cl_{1,3}}
%\end{eqnarray}
% For $\theta=\pi/2$, we have a Bell state $\Psi^{+Cl_{1,3}}$. These results can be generalized to a system with an arbitrary number of Majorana operators through any number of factors in the Clifford algebra tensor product. $Cl_{1,3}$.

We can easily verify that in our formulation the operator
\begin{eqnarray}
Q^{Cl^{+}_{1,3},M}=\frac{(H^{Cl_{1,3}^{+},M})^{1/2}\Gamma^{Cl_{1,3}^{+},M}}{4}\left(1+P^{Cl_{1,3}^{+},M}\right);
\end{eqnarray}
satisfies the relation
\begin{eqnarray}
[Q^{Cl^{+}_{1,3},M}, H^{Q^{Cl^{+}_{1,3},M}}]=0; \ \  \left(Q^{Cl^{+}_{1,3},M})\right)^{2}=0
\end{eqnarray}
and
\begin{eqnarray}
\{Q^{Cl^{+}_{1,3},M},(Q^{Cl^{+}_{1,3},M})^{\dagger}\}=2H^{Cl_{1,3}^{+},M},
\end{eqnarray}
that is, we have a supersymmetry algebra with $Q^{Cl^{+}_{1,3},M}$ the charge operator. This algebra corresponds here to that introduced by Haq and Kaufmann using Majorana operators \cite{Haq}. We should note that the element $(H^{Cl_{1,3}^{+},M})^{1/2}$ has an operator structure as described in section 2 and it is self-adjoint, so the definition is valid taking into account the spectral theorem. We will now show that the existence of supersymmetry algebra $(N = 2)$ implies characteristics of the eigenstates of Hamiltonian. Initially note that
\begin{eqnarray}
\{Q^{Cl^{+}_{1,3},M},P^{Cl_{1,3}^{+},M}\}=0; \ \ [Q^{Cl^{+}_{1,3},M},H^{Cl^{+}_{1,3},M}]=0
\end{eqnarray}
Suppose an algebraic spinor $\Psi^{Cl^{+}_{1,3},M}_{0}$ an eigenstate of $H^{Cl^{+}_{1,3},M}$ with parity $+1$. We have that
\begin{eqnarray}
\Psi^{Cl^{+}_{1,3},M}_{1}=Q^{Cl^{+}_{1,3},M}\Psi^{Cl^{+}_{1,3},M}_{0}
\end{eqnarray}
is an eigenstate of $H^{Cl^{+}_{1,3},M}$ orthogonal to $\Psi^{Cl^{+}_{1,3},M}_{0}$. In fact
\begin{eqnarray}
\left\langle \left( \Psi^{Cl^{+ }_{1,3},M   }_{0}\right)^{\star}  \ Q^{Cl^{+}_{1,3},M} \Psi^{Cl^{+}_{1,3},M}_{0} \right\rangle_{0} &=& \left\langle \left( \Psi^{Cl^{+}_{1,3},M   }_{0} \right)^{\star} \ Q^{Cl^{+}_{1,3},M} P^{Cl_{1,3}^{+},M}  \Psi^{Cl^{+}_{1,3},M}_{0} \right\rangle_{0} \nonumber \\
& = & \left\langle \left( \Psi^{Cl^{+}_{1,3},M  }_{0} \right)^{\star} \ P^{Cl_{1,3}^{+},M} Q^{Cl^{+}_{1,3},M}   \Psi^{Cl^{+}_{1,3},M}_{0} \right\rangle_{0} \nonumber \\
&=&- \left\langle \left( \Psi^{Cl^{+}_{1,3},M  }_{0} \right)^{\star} \ Q^{Cl^{+}_{1,3},M} \Psi^{Cl^{+}_{1,3},M}_{0} \right\rangle_{0} ,
\end{eqnarray}
so
\begin{eqnarray}
\left\langle \left( \Psi^{Cl^{+}_{1,3},M}_{0} \right)^{\star} Q^{Cl^{+}_{1,3},M}  \Psi^{Cl^{+}_{1,3},M}_{0} \right\rangle_{0} &=& 0,
\end{eqnarray}
taking into account that $Q^{Cl^{+}_{1,3},M}\Psi^{Cl^{+}_{1,3},M}_{0}\neq 0$ for a state with parity $1$. This show the degeneracy of the Hamiltonian.

\section{Conclusions} \label{sec6}
In this paper, we derived a formulation for q-bit multipartite systems using algebraic spinors in Clifford algebra $ Cl_ {1,3}^{+}$. Our formalism involves the tensor product of algebras as well as the concept of minimal left ideals, i.e. we derived such a formulation by constructing an algebra product generated from tensor algebra product from the quotient of a tensor algebra by a two-side ideal. The universality of this algebra has been demonstrated. We note that this prescription can be generalized to the tensor product of any Clifford algebras.  In addition we show a structure of $c^{\star}$-algebra within of $Cl_{1,3}^{+}$ algebra. We note also that the derived operators are hermitian or unitary corresponding to observables and quantum logic gates, respectively.
All states and operators are described in terms of the algebra generators. Thus it is possible to perform a more efficient operationalization, ratifying several proposals already existing in the literature.
We showed, for example, how entangled states and braid group representations, essential structures in topological quantum computing proposals can be constructed algebraically; in this context we then derived an algebraic equation for quantum teleportation. We presented Majorana operators and an associated supersymmetry algebra using our algebraic formulation and we show a relationship between supersymmetry and degeneracy of energy levels.  In short, the use of Clifford algebras and algebraic spinors provides a natural, elegant and effective path to the field of quantum computing and quantum information. Therefore, we believe that such results may be useful in future developments of quantum computing, particularly in topological quantum computing and in system simulations for relativistic quantum mechanics and high energy physics.

%%%%%%%%%%%%%%%%%%%%%%%%%%%%%%%%%%%%%%%%%%%%%%%%%%%%%%%%%%%%%%%%%%%%%%%%%%
%%%%%%%%%%%%%%%%%%%%%%%%%%%%%%%%%%%%%%%%%%%%%%%%%%%%%%%%%%%%%%%%%%%%%%%%%%

\end{document}